\documentclass[11pt]{article}
\usepackage{PRIMEarxiv}
\usepackage{amsmath,amssymb,amsfonts,amsthm,mathtools}
\usepackage{booktabs}
\usepackage[hidelinks]{hyperref}
\usepackage[nameinlink,capitalize,noabbrev]{cleveref}
\usepackage{enumitem}
\usepackage{tikz}
\usetikzlibrary{positioning}
\usepackage{graphicx}
\usepackage{caption}
\usepackage{algorithm}
\usepackage{algpseudocode}
\usepackage{fancyhdr}

\title{On the Complexity of Minimum Riesz s-Energy Subset Selection in Euclidean and Ultrametric Spaces}

\author{%
Michael T. M. Emmerich\textsuperscript{1}
\And
Ksenia Pereverdieva\textsuperscript{2}
\And
Andr\'e Deutz\textsuperscript{2}\\[0.5em]
\textsuperscript{1}Faculty of Information Technology, University of Jyv\"askyl\"a, Finland\\
\textsuperscript{2}Leiden Institute of Advanced Computer Science, Leiden University, The Netherlands
}
\date{}

\newcommand{\setcleanrunningheader}{%
  \pagestyle{fancy}%
  \fancyhf{}%
  \fancyhead[L]{\footnotesize\itshape Minimum Riesz $s$-Energy Subset Selection}%
  \fancyfoot[C]{\thepage}%
  \renewcommand{\headrulewidth}{0.4pt}%
}
\AtBeginDocument{\setcleanrunningheader}

\newtheorem{theorem}{Theorem}
\newtheorem{proposition}{Proposition}
\newtheorem{definition}{Definition}
\newtheorem{lemma}{Lemma}
\newtheorem{remark}{Remark}
\newtheorem{corollary}{Corollary}

\newcommand{\Es}[2]{E_{#1}\!\left(#2\right)}
\newcommand{\MM}[1]{\operatorname{MPD}\!\left(#1\right)}
\newcommand{\bin}[1]{\binom{#1}{2}}

\definecolor{rieszffdd55}{HTML}{FFDD55}
\definecolor{riesz5580ff}{HTML}{5580FF}
\definecolor{rieszff55af}{HTML}{FF55AF}
\definecolor{riesz670133}{HTML}{670133}
\definecolor{riesz999999}{HTML}{999999}
\definecolor{riesz4d4d4d}{HTML}{4D4D4D}
\definecolor{rieszffe680}{HTML}{FFE680}
\definecolor{rieszred}{HTML}{FF0101}

\newcommand{\rieszdisc}[3]{%
  \path[draw=black,line width=0.266pt,dash pattern=on 1pt off 1pt,fill=riesz#3,fill opacity=0.218254]
    (#1,#2) ellipse [x radius=10.488082, y radius=10.130533];%
}
\newcommand{\rieszdot}[2]{%
  \path[fill=black,draw=none] (#1,#2) ellipse [x radius=0.834279, y radius=0.893871];%
}

\newcommand{\TikzBooleanCircuit}{%
\begin{tikzpicture}[x=1pt,y=-1pt]
  \path[use as bounding box] (0,0) rectangle (210,225);
  \draw[line width=1pt] (31.568,10.067) -- (31.027,220.589) -- (45.974,220.207) -- (46.165,199.951) -- (166.171,200.333);
  \draw[line width=1pt] (165.216,195.365) -- (45.592,195.174) -- (46.356,125.425) -- (165.407,125.043);
  \draw[line width=1pt] (45.783,119.119) -- (165.598,120.265);
  \draw[line width=1pt] (46.165,118.737) -- (45.592,49.943) -- (165.407,50.134);
  \draw[line width=1pt] (45.974,44.975) -- (165.025,44.784);
  \draw[line width=1pt] (45.783,9.623) -- (45.783,44.784);

  \draw[rieszred,line width=1pt] (61.261,10.387) -- (60.879,220.398);
  \draw[rieszred,line width=1pt] (60.688,220.016) -- (75.976,220.207) -- (75.976,210.461) -- (175.726,209.697) -- (175.535,205.110);
  \draw[rieszred,line width=1pt] (170.375,205.493) -- (76.167,205.110) -- (76.167,60.071) -- (176.490,59.880) -- (176.299,54.912);
  \draw[rieszred,line width=1pt] (171.140,54.721) -- (75.593,55.294) -- (75.976,10.769);
  \draw[rieszred,line width=1pt] (105.977,10.005) -- (106.168,220.207) -- (121.073,220.207) -- (120.882,189.823) -- (170.184,190.205) -- (170.757,189.823);
  \draw[rieszred,line width=1pt] (176.681,190.396) -- (176.299,185.237) -- (120.882,185.237) -- (121.073,134.979) -- (175.917,135.362) -- (175.917,130.393);
  \draw[rieszred,line width=1pt] (171.140,130.011) -- (120.691,129.629) -- (120.691,10.769);
  \draw[rieszred,line width=1pt] (1.129,10.652) -- (0.502,220.697) -- (15.983,220.488) -- (16.401,115.256) -- (170.797,115.256);
  \draw[rieszred,line width=1pt] (175.817,115.047) -- (175.608,110.863);
  \draw[rieszred,line width=1pt] (175.399,110.444) -- (16.401,109.817) -- (15.774,40.151) -- (170.797,40.151);
  \draw[rieszred,line width=1pt] (175.608,39.732) -- (175.608,34.711) -- (15.983,35.130) -- (15.564,10.861);

  \draw[line width=1pt] (76.570,167.575) -- (76.152,162.764) -- (79.290,162.973) -- (78.871,156.906) -- (76.570,156.906) -- (76.152,152.512);
  \draw[line width=1pt] (121.346,168.090) -- (120.928,163.278) -- (124.066,163.487) -- (123.648,157.420) -- (121.346,157.420) -- (120.928,153.027);
  \draw[line width=1pt] (106.394,167.683) -- (105.976,162.872) -- (109.114,163.081) -- (108.696,157.014) -- (106.394,157.014) -- (105.976,152.620);

  \node[anchor=west,font=\bfseries\scriptsize,text=rieszred] at (168.914,49.356) {Clause Point};
  \node[anchor=west,font=\bfseries\scriptsize,text=rieszred] at (169.246,124.437) {Clause Point};
  \node[anchor=west,font=\bfseries\scriptsize,text=rieszred] at (168.797,200.755) {Clause Point};
  \node[anchor=west,font=\bfseries\scriptsize,text=rieszred] at (5.313,5.631) {x1};
  \node[anchor=west,font=\bfseries\scriptsize] at (35.230,5.213) {x2};
  \node[anchor=west,font=\bfseries\scriptsize,text=rieszred] at (64.937,5.631) {x3};
  \node[anchor=west,font=\bfseries\scriptsize,text=rieszred] at (109.917,5.311) {xn};
\end{tikzpicture}%
}

\newcommand{\TikzWireGadget}{%
\begin{tikzpicture}[x=1pt,y=-1pt]
  \path[use as bounding box] (0,0) rectangle (76.276,64.572);
  \rieszdisc{25.846}{13.420}{670133}\rieszdot{25.803}{13.360}
  \rieszdisc{50.155}{14.090}{670133}\rieszdot{50.112}{14.029}
  \rieszdisc{20.396}{24.334}{ff55af}\rieszdisc{44.597}{53.499}{ff55af}
  \rieszdot{44.555}{53.439}\rieszdot{20.353}{24.273}
  \rieszdisc{63.802}{39.960}{ff55af}\rieszdot{63.759}{39.899}
  \rieszdisc{38.429}{10.263}{ff55af}\rieszdot{38.387}{10.203}
  \rieszdisc{62.321}{16.808}{ff55af}\rieszdot{62.279}{16.748}
  \rieszdisc{65.654}{28.054}{670133}\rieszdot{65.612}{27.994}
  \rieszdisc{56.054}{49.568}{670133}\rieszdot{56.012}{49.508}
  \rieszdisc{32.767}{53.004}{670133}\rieszdot{32.725}{52.944}
  \rieszdisc{18.051}{36.045}{670133}\rieszdot{18.009}{35.985}
  \rieszdisc{21.664}{48.821}{ff55af}\rieszdot{21.622}{48.761}
\end{tikzpicture}%
}

\newcommand{\TikzCrossoverGadget}{%
\begin{tikzpicture}[x=1pt,y=-1pt]
  \path[use as bounding box] (0,0) rectangle (65.619,65.171);
  \rieszdisc{46.739}{18.559}{ffe680}\rieszdot{46.612}{18.583}
  \rieszdisc{11.005}{54.908}{ffe680}\rieszdot{10.878}{54.932}
  \rieszdisc{54.727}{10.361}{ffdd55}\rieszdot{54.600}{10.385}
  \rieszdisc{18.666}{18.491}{999999}\rieszdot{18.540}{18.515}
  \rieszdisc{26.488}{38.448}{5580ff}\rieszdot{26.446}{38.388}
  \rieszdisc{38.578}{38.507}{5580ff}\rieszdot{38.536}{38.446}
  \rieszdisc{39.023}{26.493}{5580ff}\rieszdot{38.981}{26.432}
  \rieszdisc{26.878}{26.301}{5580ff}\rieszdot{26.835}{26.240}
  \rieszdisc{18.389}{46.918}{ffdd55}\rieszdot{18.515}{46.894}
  \rieszdisc{54.998}{54.617}{999999}\rieszdot{55.125}{54.593}
  \rieszdisc{10.621}{10.264}{4d4d4d}\rieszdot{10.494}{10.288}
  \rieszdisc{46.989}{46.176}{4d4d4d}\rieszdot{46.862}{46.200}
\end{tikzpicture}%
}

\newcommand{\TikzClauseGadget}{%
\begin{tikzpicture}[x=1pt,y=-1pt]
  \path[use as bounding box] (0,0) rectangle (72.971,75.065);
  \rieszdisc{22.009}{51.988}{ffe680}\rieszdot{21.882}{52.012}
  \rieszdisc{21.725}{64.801}{ffdd55}\rieszdot{21.598}{64.825}
  \rieszdisc{62.350}{42.667}{999999}\rieszdot{62.223}{42.691}
  \rieszdisc{36.448}{34.936}{5580ff}\rieszdot{36.405}{34.876}
  \rieszdisc{40.645}{41.674}{5580ff}\rieszdot{40.602}{41.614}
  \rieszdisc{32.238}{41.540}{5580ff}\rieszdot{32.195}{41.479}
  \rieszdisc{10.621}{42.495}{ffdd55}\rieszdot{10.748}{42.471}
  \rieszdisc{50.876}{64.757}{999999}\rieszdot{51.003}{64.733}
  \rieszdisc{51.074}{52.037}{4d4d4d}\rieszdot{50.948}{52.061}
  \rieszdisc{24.740}{18.540}{670133}\rieszdot{24.698}{18.479}
  \rieszdisc{48.832}{18.789}{670133}\rieszdot{48.789}{18.728}
  \rieszdisc{16.585}{10.263}{ff55af}\rieszdot{16.542}{10.203}
  \rieszdisc{36.577}{18.666}{ff55af}\rieszdot{36.535}{18.606}
  \rieszdisc{56.765}{10.395}{ff55af}\rieszdot{56.723}{10.334}
\end{tikzpicture}%
}

\newcommand{\TikzHexagonalPacking}{%
\begin{tikzpicture}[x=1pt,y=-1pt]
  \path[use as bounding box] (0,0) rectangle (93.898,102.160);
  \rieszdisc{46.849}{71.735}{ffdd55}\rieszdot{46.722}{71.759}
  \rieszdisc{46.707}{91.896}{ffdd55}\rieszdot{46.580}{91.920}
  \rieszdisc{10.761}{71.493}{ffdd55}\rieszdot{10.634}{71.517}
  \rieszdisc{28.705}{81.166}{ffdd55}\rieszdot{28.579}{81.190}
  \rieszdisc{10.621}{51.025}{ffdd55}\rieszdot{10.494}{51.049}
  \rieszdisc{10.761}{30.557}{ffdd55}\rieszdot{10.634}{30.581}
  \rieszdisc{46.851}{51.310}{5580ff}\rieszdot{46.809}{51.250}
  \rieszdisc{82.709}{72.169}{ffdd55}\rieszdot{82.836}{72.145}
  \rieszdisc{64.876}{61.715}{ffdd55}\rieszdot{65.003}{61.691}
  \rieszdisc{65.138}{40.758}{ffdd55}\rieszdot{65.265}{40.735}
  \rieszdisc{28.809}{40.754}{ffdd55}\rieszdot{28.936}{40.730}
  \rieszdisc{47.285}{30.745}{ffdd55}\rieszdot{47.411}{30.721}
  \rieszdisc{28.845}{20.272}{ffdd55}\rieszdot{28.718}{20.296}
  \rieszdisc{46.926}{10.263}{ffdd55}\rieszdot{46.799}{10.287}
  \rieszdisc{65.827}{20.147}{ffdd55}\rieszdot{65.953}{20.123}
  \rieszdisc{83.041}{31.251}{ffdd55}\rieszdot{83.168}{31.227}
  \rieszdisc{83.277}{51.728}{ffdd55}\rieszdot{83.150}{51.752}
  \rieszdisc{28.535}{61.336}{ffdd55}\rieszdot{28.408}{61.360}
  \rieszdisc{64.485}{82.048}{ffdd55}\rieszdot{64.612}{82.024}
\end{tikzpicture}%
}

\begin{document}
\maketitle
\setcleanrunningheader

\begin{abstract}
We study the computational complexity of exact cardinality-constrained minimum Riesz $s$-energy subset selection in finite metric spaces: given $n$ points, select $k<n$ points of minimum Riesz $s$-energy. The objective sums inverse-power pair interactions and therefore promotes well-separated subsets; as $s$ becomes large, it increasingly approaches a bottleneck criterion governed by the closest selected pair, linking it to minimum pairwise distance (MPD). Building on the general-metric NP-hardness result of Pereverdieva et al. (2025), we prove that NP-hardness persists for point sets in the Euclidean plane when $s$ is part of the input. In contrast, finite ultrametric spaces form an exact tractable regime: on rooted binary ultrametric trees with $n$ leaves, an optimal size-$k$ subset can be computed by dynamic programming in $O(nk^2)$ time. We also discuss the ordered one-dimensional Euclidean case, where the classical MPD objective admits simple dynamic programming, but the additive Riesz energy does not appear to allow the same state compression. Finally, we explain why one natural route to fixed-$s$ Euclidean hardness does not close: Fowler-style \textsc{3SAT} gadgets, together with zeta-function bounds for far-field interactions, show why this approach still requires an exponent depending on $k$. Together, these results provide a compact complexity landscape for a natural diversity or dispersion objective, distinguishing Euclidean hardness, ultrametric tractability, and the ordered one-dimensional case.
\end{abstract}


\section{Introduction}

Riesz energies are a classical family of pairwise interaction objectives. For exponent $s>0$, nearby points interact strongly through the term $d(x,y)^{-s}$, and minimizing the total interaction encourages well-separated configurations. This viewpoint is central in potential theory and in the study of energy-minimizing point distributions on spheres and manifolds, where it is shown to lead to uniformly spaced configurations; see, for example, the survey of Brauchart and Grabner~\cite{BrauchartGrabner2015} and the reference work \cite{BorodachovHardinSaff2019}. Beyond this classical setting, minimum-subset selection for discrete Riesz $s$-energy has recently received broader attention in applied multiobjective optimization and diversity maximization, where it is used to extract well-distributed finite Pareto-front approximations and to compare diversity indicators, and it has also appeared in closely related geometric optimization problems in theoretical computer science~\cite{Atta2026,FalconCardonaUribeRosas2024,PereverdievaEtAl2025}. 

A useful perspective is that, on every fixed finite candidate set, minimum
Riesz $s$-energy interpolates between pairwise repulsion and minimum pairwise distance
dispersion. Indeed, as $s\to\infty$, the contribution of the closest selected
pair dominates the sum, so minimizing Riesz $s$-energy tends to maximizing the
minimum pairwise distance. In this sense, the MPD optimization problem
appears as the large-$s$ limit of minimum Riesz $s$-energy subset selection;
see Chapter~3 of Borodachov, Hardin, and
Saff~\cite{BorodachovHardinSaff2019} and also Atta~\cite{Atta2026}.

Max--min diversity has been studied for decades, not only in abstract metric
spaces but also in geometric settings, where hardness arguments are closely
related to geometric independent set and packing constructions~\cite{ghosh1996computational,RaviRosenkrantzTayi1994,FowlerPatersonTanimoto1981}. Our Euclidean-plane hardness proof belongs to this line of work. At the same time, the Riesz objective is structurally subtler: unlike minimum pairwise distance (MPD), it depends on the full collection of pairwise interactions rather than only on the closest selected pair.

In this paper we study the cardinality-constrained subset-selection problem for the discrete Riesz $s$-energy. If $S$ is a $k$-subset of a metric space, then the objective aggregates all pairwise interactions $d(x,y)^{-s}$ over distinct selected points. Thus it differs from the maximum-sum dispersion objective, which directly rewards large distances, and from the classical MPD objective, which is determined solely by the closest selected pair. At the same time, for large $s$ the Riesz objective increasingly concentrates on the smallest selected distance. In that sense it forms a natural bridge between pairwise-sum repulsion and minimum pairwise distance (MPD).

Our main theme is a sharp contrast between hard general metric cases and a tractable hierarchical regime. A general-metric NP-hardness result was already established by Pereverdieva et al.~\cite{PereverdievaEtAl2025}. Building on that hard side of the landscape, we show that NP-hardness persists in the Euclidean plane when $s$ is part of the input. By contrast, ultrametric spaces admit an exact dynamic programming algorithm. This tractable case is still nontrivial: one must distribute a fixed budget across competing branches of a hierarchy under a global objective, but the ultrametric structure collapses every cross-subtree interaction to a simple cardinality-dependent term.

The paper is deliberately selective. Rather than attempting a full survey of all geometric and parameterized variants, we isolate one coherent picture: hardness in general metrics, exact tractability on ultrametrics, and hardness in Euclidean spaces. Throughout the study, we relate Riesz energy to minimum pairwise distance (MPD), treating it as a limiting case.. This makes the theorem flow short and leaves the unresolved directions visible rather than buried.

\paragraph{Contributions.}
The paper has five principal contributions, one of which is contextual.
\begin{itemize}[leftmargin=1.5em]
\item For general finite metric spaces, we reproduce and refine the NP-hardness result of Pereverdieva et al.~\cite{PereverdievaEtAl2025}. For completeness and in order to keep the paper self-contained, we include a short hardness proof in our notation and make the reduction argument explicit.
\item We prove Euclidean NP-hardness in the plane when $s$ is part of the input.
\item We identify finite ultrametric spaces as a tractable island by deriving an exact dynamic programming algorithm running in $O(nk^2)$ time on rooted binary ultrametric trees.
\item We record the one-dimensional comparison with minimum pairwise distance (MPD): on the line, the classical MPD objective admits a Bellman-style dynamic program, whereas no analogous state compression is available for minimum Riesz $s$-energy. This clarifies why the one-dimensional Euclidean case remains structurally delicate.
\item For the minimum pairwise distance (MPD) comparison, we use the known finite-instance fact that for sufficiently large $s$, every Riesz-energy minimizer is MPD, and we highlight the implications of this limiting case for hardness and tractability.
\end{itemize}
These results leave a short list of structurally meaningful open cases, most notably fixed-exponent Euclidean hardness.

\section{Preliminaries}

\begin{definition}[Minimum Riesz $s$-energy subset selection]
Let $(X,d)$ be a finite metric space and let $s>0$ be fixed. For a subset $S\subseteq X$ with $|S|\ge 2$, define the unordered-pair Riesz $s$-energy by
\[
\Es{s}{S}=\sum_{\{u,v\}\subseteq S} d(u,v)^{-s}.
\]
We also set $E_s(S)=0$ when $|S|\le 1$. The optimization problem asks for a $k$-subset minimizing $E_s(S)$.
\end{definition}

We will use the associated decision problem.

\medskip
\noindent
\textbf{RSSP-Decision:} Given $(X,d)$, a target size $k$, an exponent $s>0$, and a threshold $T$, decide whether there exists a subset $S\subseteq X$ with $|S|=k$ and $E_s(S)\le T$.

\medskip
We also use the classical minimum pairwise distance (MPD) objective
\[
\MM{S}:=\min_{\{u,v\}\subseteq S} d(u,v),
\]
with the convention that $\MM{S}=+\infty$ for $|S|\le 1$.

\begin{definition}[Ultrametric]
A metric $d$ on a set $X$ is called an \emph{ultrametric} if
\[
d(x,z)\le \max\{d(x,y),d(y,z)\}
\qquad\text{for all }x,y,z\in X.
\]
Equivalently, for every triple of points, the two largest pairwise distances are equal.
\end{definition}

Every finite ultrametric admits a rooted-tree representation: the points are the leaves, internal nodes carry nondecreasing heights, and the distance between two leaves is twice the height of their least common ancestor. Ultrametrics arise naturally from rooted phylogenetic trees, where distances may represent species separation in evolutionary time, and from the cophenetic distances of hierarchical clustering dendrograms, including hierarchical analyses of cultural vectors~\cite{HartmannSteel2006,StivalaEtAl2014}.

\section{General metric spaces via \texorpdfstring{$k$-Clique}{k-Clique}}

We start with the purely metric setting. The corresponding NP-hardness statement already appears in the comparative analysis of Pereverdieva et al.~\cite{PereverdievaEtAl2025}; we include a short self-contained proof because it serves as the clean baseline for the rest of the paper.

\begin{definition}[$k$-Clique]
Given a graph $G=(V,E)$ and an integer $k$, decide whether there exists a subset $V'\subseteq V$ of size $k$ such that every pair of distinct vertices in $V'$ is joined by an edge.
\end{definition}

\begin{lemma}\label{lem:metric}
Let $G=(V,E)$. Define $d:V\times V\to\{0,1,2\}$ by
\[
d(u,v)=\begin{cases}
0,&u=v,\\
2,&u\ne v\text{ and }\{u,v\}\in E,\\
1,&u\ne v\text{ and }\{u,v\}\notin E.
\end{cases}
\]
Then $(V,d)$ is a metric space.
\end{lemma}

\begin{proof}
Non-negativity, symmetry, and $d(u,v)=0\Leftrightarrow u=v$ are immediate. For the triangle inequality, the only nontrivial case is $d(u,v)=2$. If $u\neq w\neq v$, then each of $d(u,w)$ and $d(w,v)$ belongs to $\{1,2\}$, so $2\le d(u,w)+d(w,v)$. All other cases are trivial.
\end{proof}

\begin{theorem}\label{thm:metric}
For every fixed $s>0$, RSSP-Decision is NP-hard for general finite metric spaces.
\end{theorem}

\begin{proof}
Given $(G,k)$, build $(V,d)$ as in \cref{lem:metric}. For any $k$-subset $S$,
\[
\bin{k}\cdot 2^{-s}\le E_s(S)\le \bin{k}.
\]
Set $T:=\bin{k}\cdot 2^{-s}$. Then $G$ has a $k$-clique if and only if there exists a $k$-subset $S$ with all pairwise distances equal to $2$, which holds if and only if $E_s(S)=T\le T$. This is a polynomial reduction from $k$-Clique. The argument uses only the strict distance gap between $1$ and $2$, and therefore does not require $s\ge 1$. For fixed positive integer $s$, all numbers in this reduction are rational; for a fixed non-integer exponent, the constant $2^{-s}$ is part of the fixed exact-real problem specification.
\end{proof}

\begin{remark}
Theorem~\ref{thm:metric} is not the main novelty of the present paper. Its role is to anchor the hard side of the complexity landscape in a self-contained way, while the general-metric NP-hardness result itself already appears in the comparative analysis of Pereverdieva et al.~\cite{PereverdievaEtAl2025}.
\end{remark}

\section{Euclidean hardness in the plane when \texorpdfstring{$s$}{s} is part of the input}

We now restrict to Euclidean geometry and reduce geometric independent set to
RSSP using a global threshold based on the gap between forbidden and
admissible distances; see Figure~\ref{fig:dmin-dmax}.

\begin{definition}[Geometric Independent Set]
Given a finite point set $P\subset\mathbb{R}^2$, a threshold $\delta>0$, and an integer $k$, decide whether there exists a subset $S\subseteq P$ with $|S|=k$ such that $\|x-y\|\ge \delta$ for all distinct $x,y\in S$.
\end{definition}

This problem is NP-complete in the plane~\cite{ClarkColbournJohnson1990,AgarwalVanKreveldSuri1998,FowlerPatersonTanimoto1981}. For a finite set $X\subseteq P$, define
\[
(X\times X)_{\ge \delta}:=\{(x,y):x,y\in X,\ x\neq y,\ \|x-y\|\ge \delta\},
\]
\[
(X\times X)_{< \delta}:=\{(x,y):x,y\in X,\ x\neq y,\ \|x-y\|< \delta\}.
\]
We then set
\[
D_{\min}(X):=
\begin{cases}
\min\{\|x-y\|:(x,y)\in (X\times X)_{\ge\delta}\}, & \text{if }(X\times X)_{\ge\delta}\neq\emptyset,\\
+\infty, & \text{otherwise,}
\end{cases}
\]
and
\[
\delta_{\max}(X):=
\begin{cases}
\max\{\|x-y\|:(x,y)\in (X\times X)_{<\delta}\}, & \text{if }(X\times X)_{<\delta}\neq\emptyset,\\
-\infty, & \text{otherwise.}
\end{cases}
\]
For the ambient point set $P$, we abbreviate
\[
D_{\min}:=D_{\min}(P),\qquad \delta_{\max}:=\delta_{\max}(P).
\]
Figure~\ref{fig:dmin-dmax} illustrates these quantities. The key point is that
we need a strict gap between the largest forbidden distance $\delta_{\max}$
and the smallest admissible distance $D_{\min}$.

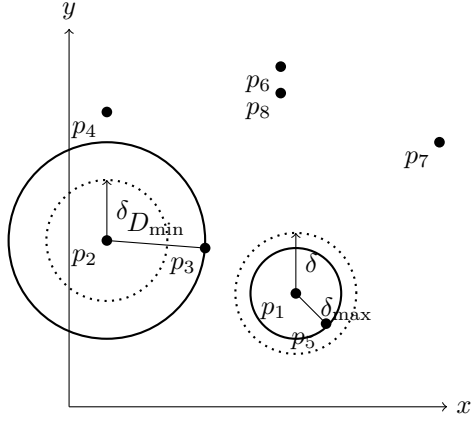
\begin{figure}[t]
  \centering
  \begin{minipage}[t]{0.62\linewidth}
    \centering
    \begin{tikzpicture}[scale=1.0]
      \draw[->] (0,0) -- (5,0) node[right] {$x$};
      \draw[->] (0,0) -- (0,5) node[above] {$y$};
      \draw (3,1.5) -- (3.4,1.1) node[midway, right] {$\delta_{\max}$};
      \draw (0.5,2.2) -- (1.8,2.1) node[midway, above] {$D_{\min}$};
      \draw[thick] (3,1.5) circle (0.6cm);
      \draw[thick,dotted] (3,1.5) circle (0.8cm);
      \draw[thick] (0.5,2.2) circle (1.3cm);
      \draw[thick,dotted] (0.5,2.2) circle (0.8cm);
      \fill (3,1.5) circle (2pt) node[below left] {$p_1$};
      \draw[->](3,1.5) -- (3,2.3) node[midway, right] {$\delta$};
      \fill (0.5,2.2) circle (2pt) node[below left] {$p_2$};
      \draw[->](0.5,2.2) -- (0.5,3.0) node[midway, right] {$\delta$};
      \fill (1.8,2.1) circle (2pt) node[below left] {$p_3$};
      \fill (0.5,3.9) circle (2pt) node[below left] {$p_4$};
      \fill (3.4,1.1) circle (2pt) node[below left] {$p_5$};
      \fill (2.8,4.5) circle (2pt) node[below left] {$p_6$};
      \fill (4.9,3.5) circle (2pt) node[below left] {$p_7$};
      \fill (2.8,4.15) circle (2pt) node[below left] {$p_8$};
    \end{tikzpicture}
    \caption{Illustration of $\delta_{\max}$ and $D_{\min}$ used in the GIS-to-RSSP reduction. The dotted circles indicate the threshold radius $\delta$ around $p_1$ and $p_2$. Here $\delta_{\max}$ is realized by the forbidden pair of maximum radius $(p_1,p_5)$ lying just above the threshold, while $D_{\min}$ is realized by the admissible pair with minimum radius $(p_2,p_3)$ lying at or above the threshold. Choosing $s$ so that $\binom{k}{2}D_{\min}^{-s}<\delta_{\max}^{-s}$ ensures that one forbidden pair already contributes more energy than all admissible pairs together.}
    \label{fig:dmin-dmax}
  \end{minipage}
\end{figure}

\begin{lemma}\label{lem:dmin-dmax-monotone}
Let $\emptyset\neq S\subseteq P$. Then
\[
\delta_{\max}(S)\le \delta_{\max}(P)<\delta\le D_{\min}(P)\le D_{\min}(S).
\]
In particular, if both forbidden and admissible pairs occur in $P$, then
\[
\delta_{\max}<D_{\min}.
\]
\end{lemma}

\begin{proof}
The inclusions
\[
(S\times S)_{<\delta}\subseteq (P\times P)_{<\delta},
\qquad
(S\times S)_{\ge\delta}\subseteq (P\times P)_{\ge\delta}
\]
immediately imply
\[
\delta_{\max}(S)\le \delta_{\max}(P),
\qquad
D_{\min}(P)\le D_{\min}(S),
\]
with the stated conventions when one of the sets is empty. By definition,
every element of $(P\times P)_{<\delta}$ has distance strictly smaller than
$\delta$, and every element of $(P\times P)_{\ge\delta}$ has distance at least
$\delta$. Hence, because $P$ is finite, 
\[
\delta_{\max}(P)<\delta\le D_{\min}(P).
\]
Combining the inequalities yields the claim.
\end{proof}

\begin{lemma}[Auxiliary exponent choice]\label{lem:aux-exp-choice}
Let $c>0$ and let $0<a<b$. Then the inequality
\[
cb^{-u}<a^{-u}
\]
holds precisely for
\[
u>\frac{\log(c)}{\log(b)-\log(a)}.
\]
In particular, if $u_0$ is any real number with
\[
u_0>\frac{\log(c)}{\log(b)-\log(a)},
\]
then $cb^{-u}<a^{-u}$ for every $u\ge u_0$.
\end{lemma}

\begin{proof}
Since $0<a<b$, we have $\log(b)-\log(a)>0$. Taking logarithms, the inequality
\[
cb^{-u}<a^{-u}
\]
is equivalent to
\[
\log(c)<u(\log(b)-\log(a)),
\]
which is exactly the stated condition. The second claim is immediate.
\end{proof}

Before stating the main NP-hardness theorem of this section, we briefly outline the coding conventions for the Euclidean case:

\begin{remark}[Encoding convention for the Euclidean reduction]
For the following reduction we use the usual exact encoding for Euclidean instances in computational geometry. The input point coordinates and the threshold $\delta$ of the geometric independent-set instance are rational numbers encoded in binary. Squared Euclidean distances are therefore rational, and distances such as $D_{\min}$ and $\delta_{\max}$ are represented exactly by the corresponding realizing pairs and the square roots of their rational squared distances. In the constructed RSSP instance, the exponent $s$ is a positive integer encoded in binary, and the threshold is represented by the straight-line expression
\(
T=\binom{k}{2}D_{\min}^{-s}.
\).
Exponentiation is stored succinctly, for instance by repeated squaring, rather than by expanding numerators and denominators. If $L$ is the bit length of the geometric input, standard separation bounds for distinct rational squared distances imply that the integer $s$ chosen below has polynomial bit length in $L$. Thus the constructed RSSP instance has polynomial encoding length. No decimal approximation of $D_{\min}$, $\delta_{\max}$, $s$, or $T$ is used in the reduction.
\end{remark}

We can now state the central theorem for the Euclidean plane:
\begin{theorem}\label{thm:gis}
RSSP-Decision in the Euclidean plane is NP-hard when $s$ is part of the input.
\end{theorem}

\begin{proof}
Let $(P,\delta,k)$ be an instance of geometric independent set. If either
$(P\times P)_{<\delta}=\emptyset$ or $(P\times P)_{\ge\delta}=\emptyset$, then
the instance is trivial and can be handled directly. We therefore assume that
both forbidden and admissible pairs occur. By
Lemma~\ref{lem:dmin-dmax-monotone}, we have $\delta_{\max}<D_{\min}$.

By comparing squared distances, compute in polynomial time the realizing pairs for $\delta_{\max}$ and $D_{\min}$, and represent these distances exactly as described above. Choose the positive integer
\[
s:=1+\left\lceil\frac{\log \bin{k}}{\log(D_{\min}/\delta_{\max})}\right\rceil,
\qquad
T:=\bin{k}D_{\min}^{-s}.
\]
The exponent $s$ is encoded in binary and $T$ is encoded by the displayed expression. By Lemma~\ref{lem:aux-exp-choice} with $c=\bin{k}$, $a=\delta_{\max}$, and $b=D_{\min}$, this choice guarantees
\[
\bin{k}D_{\min}^{-s}<\delta_{\max}^{-s}.
\]
We claim that for every $k$-subset $S\subseteq P$:
\begin{enumerate}
\item if $S$ is $\delta$-independent, then $E_s(S)\le T$;
\item if $S$ is not $\delta$-independent, then $E_s(S)>T$.
\end{enumerate}

If $S$ is $\delta$-independent, then every selected pair lies in
$(S\times S)_{\ge\delta}$, so every selected distance is at least
$D_{\min}(S)\ge D_{\min}$ by Lemma~\ref{lem:dmin-dmax-monotone}. Hence
\[
E_s(S)\le \bin{k}D_{\min}^{-s}=T.
\]
Conversely, if $S$ is not $\delta$-independent, then some selected pair has
distance strictly smaller than $\delta$, and therefore at most
$\delta_{\max}(S)\le \delta_{\max}$ by
Lemma~\ref{lem:dmin-dmax-monotone}. Hence
\[
E_s(S)\ge \delta_{\max}^{-s}.
\]
Therefore $T<\delta_{\max}^{-s}$. Hence every non-independent $k$-subset has
energy strictly larger than $T$.

Thus the GIS instance is a YES-instance if and only if the constructed RSSP
instance is a YES-instance.
\end{proof}

\begin{remark}
The essential point in Theorem~\ref{thm:gis} is the global threshold: the
ambient gap between $\delta_{\max}$ and $D_{\min}$, illustrated in
Figure~\ref{fig:dmin-dmax}, is converted into an energy separation by choosing
$s$ large enough that a single forbidden distance dominates the total
contribution of all admissible pairs. This is also the mechanism by which the
reduction points toward the MPD limit.
\end{remark}

\begin{remark}
Since $\mathbb{R}^2$ embeds isometrically into $\mathbb{R}^d$ for every fixed
$d\ge 2$ via $(x_1,x_2)\mapsto (x_1,x_2,0,\dots,0)$, Theorem~\ref{thm:gis}
immediately implies that RSSP-Decision is NP-hard in every fixed Euclidean
dimension $d\ge 2$ when $s$ is part of the input.
\end{remark}

We close this section by discussing three obstructions that arise when one tries to remove $s$ from the input. These show that several standard routes are insightful, but ultimately do not work.

\begin{remark}
We discuss a direct reduction to \textsc{3SAT} via Fowler's gadget
\cite{FowlerPatersonTanimoto1981} in Appendix~B. This route is instructive
because it shows that one need not use a completely naive separation argument
in which a single forbidden pair must dominate the total contribution of all
$\binom{k}{2}$ admissible-pair contributions in a $\delta$-independent
$k$-subset. Owing to the constant geometric gap in that construction --- for
instance, a separation of at least $2/3$ between below-threshold and
above-threshold distances --- the exponent $s$ can be chosen to grow only
linearly with $k$, rather than proportionally to $\binom{k}{2}$. Thus the
construction shows that the exponent need not compensate for the full
quadratic number of pair contributions in an admissible subset. However, this
still does not yield a fixed-$s$ hardness result. Although the required
exponent grows only linearly with $k$, it remains tied to the subset size and
hence to the input instance. Therefore this route improves substantially over
the naive quadratic estimate, but still requires choosing $s$ as part of the
construction rather than keeping it uniformly constant.
\end{remark}

\begin{remark}
One might want to use a scaled distance matrix to achieve separation between
admissible and non-admissible cases for fixed $s$. The main obstruction is the
homogeneity of the Riesz kernel: \(E_s(\lambda X)=\lambda^{-s}E_s(X)\).
Thus, if all relevant distances are scaled by the same factor \(\lambda\),
then the energies of admissible and non-admissible subsets are both
multiplied by the same factor \(\lambda^{-s}\). Hence uniform scaling does
not improve their relative separation. In particular, a distance pattern
\((\lambda,2\lambda)\) yields only the constant factor \(2^s\) between one
forbidden pair and one admissible pair. More precisely, if in an admissible
subset all pairs have distance at least \(2\lambda\), whereas a non-admissible
subset contains at least one pair at distance at most \(\lambda\), then one
forbidden pair contributes at least \(\lambda^{-s}\), while an admissible
\(k\)-subset may still have total energy as large as
\(\binom{k}{2}(2\lambda)^{-s}\). Thus one would need
\(2^s>\binom{k}{2}\)
for a single forbidden pair to dominate the full energy of every admissible
\(k\)-subset. For fixed \(s\), this fails once \(k\) is sufficiently large.
Therefore such a simple \((\lambda,2\lambda)\) scaling cannot yield the
required separation for fixed \(s\).
\end{remark}

\begin{remark}[Heuristic obstruction to broadcast amplification]
\label{rem:broadcast-obstruction}
A plausible route to fixed-$s$ hardness would be to replicate ``check''
gadgets in additional geometric layers, with the aim that one incompatible
logical pair should trigger many short distances. However, if the replicated
check points are fixed in the instance and are not themselves selected, then
their contribution to the Riesz $s$-energy factorizes. Indeed, if the selected
representatives are $x_u$ and $y_v$ and $Z$ is any fixed auxiliary set, then
\[
E(u,v)=\sum_{z\in Z}\|x_u-z\|^{-s}
+\sum_{z\in Z}\|y_v-z\|^{-s}
+\|x_u-y_v\|^{-s}
+\mathrm{const}.
\]
Hence, the auxiliary layers contribute only unary terms depending on $u$ and
$v$ separately; they do not create new pair-specific penalties beyond the
direct interaction $\|x_u-y_v\|^{-s}$. Therefore, merely stacking fixed check
layers in a third dimension does not by itself amplify a single bad logical
pair into many independent bad terms. Any successful fixed-$s$ construction
would need an additional mechanism that forces the co-selection of companion
points, or else must encode the amplification directly in the distances among
the selected points themselves.
\end{remark}

Taken together, these three obstructions indicate that eliminating \(s\) from the input is a genuinely nontrivial challenge: while each route is conceptually informative, none presently leads to a fixed-\(s\) hardness construction in Euclidean space.

\section{Ultrametrics and exact dynamic programming}

Ultrametrics are metrics induced by rooted hierarchies. Every finite ultrametric space admits a rooted-tree representation: the points are the leaves, internal nodes carry nondecreasing heights, and the distance between two leaves is determined by the height of their least common ancestor. Such spaces arise naturally from rooted phylogenetic trees in biodiversity applications, where distances encode species separation in evolutionary time, and from the cophenetic distances of hierarchical clustering dendrograms, including hierarchical analyses of cultural vectors~\cite{HartmannSteel2006,StivalaEtAl2014}. They therefore form a natural structured class in which one may ask whether minimum Riesz $s$-energy subset selection remains hard or becomes tractable.

A recent exact algorithmic contribution is the pruning algorithm of
\cite{MinhKlaereHaeseler2006}, which computes a size-$k$ subset of maximum
phylogenetic diversity (PD) on a phylogenetic tree. In the specifically
ultrametric setting, exact algorithmic results appear to be comparatively sparse;
among the verified references included here, \cite{MansonSempleSteel2022}
gives polynomial-time methods for counting all size-$k$ maximum-PD sets on rooted
ultrametric trees and for optimizing a linear score over those maximum-PD sets.
A broader exact ultrametric framework is provided by
\cite{ArenasMerklPichlerRiveros2024}, who prove polynomial-time solvability for
selecting a size-$k$ subset maximizing any weakly subset-monotone diversity
function extending an ultrametric; this applies in particular to phylogenetic
diversity in the tree setting, as well as to the pairwise-sum diversity
$\delta_{\mathrm{sum}}$, the minimum-pairwise-distance diversity
$\delta_{\min}$, and Weitzman diversity $\delta_W$.

Also for Riesz $s$-energy subset selection, we also show that finite ultrametric spaces form an exact, tractable island. The key structural fact is that, at any internal node of the rooted tree, all cross-subtree distances are equal. Hence the full interaction between two child subproblems collapses to a single cardinality-dependent cross term.

\begin{theorem}[Ultrametric recurrence]\label{thm:ultra}
Let $F_u(t)$ denote the minimum Riesz $s$-energy among all subsets of size $t$ chosen from the leaves below node $u$ in a rooted binary ultrametric tree. Then for a leaf $\ell$,
\[
F_\ell(0)=0,\qquad F_\ell(1)=0,
\]
and for an internal binary node $u$ with children $v$ and $w$,
\[
F_u(t)=\min_{t_v+t_w=t}\bigl(F_v(t_v)+F_w(t_w)+t_v t_w\,\Delta_u^{-s}\bigr),
\]
where $\Delta_u$ is the common distance between any leaf in the subtree of $v$ and any leaf in the subtree of $w$.
\end{theorem}

\begin{corollary}\label{cor:ultra-complexity}
On a rooted binary ultrametric tree with $n$ leaves, minimum Riesz $s$-energy subset selection can be solved exactly in $O(nk^2)$ time and $O(nk)$ space.
\end{corollary}

\begin{remark}
Theorems~\ref{thm:ultra} and Corollary~\ref{cor:ultra-complexity} identify finite ultrametric spaces as the main exact tractable regime isolated in this paper. A fuller derivation of the recurrence, an explicit dynamic-programming algorithm, a biological worked example, and a brief numerical validation against brute force are deferred to Appendix~\ref{app:ultra}.
\end{remark}

\section{The large-\texorpdfstring{$s$}{s} limit and minimum pairwise distance (MPD)}

For the comparison with minimum pairwise distance (MPD), we recall the following known finite-instance observation; see, for example, Atta~\cite{Atta2026}.

\begin{proposition}[Eventual minimum pairwise distance equivalence]\label{thm:limit}
Let $(X,d)$ be a finite metric space and fix $k$. Then there exists $s_0>0$ such that for every $s\ge s_0$, every minimizer of $E_s$ over all $k$-subsets of $X$ is also a maximizer of the MPD objective $\MM{S}$.
\end{proposition}

\begin{proof}
We use a similar reasoning than in the proof of theorem \ref{thm:gis}.
Let $\mathcal{S}$ be the finite family of all $k$-subsets of $X$, and let
\[
D^*:=\max\{\MM{S}:S\in\mathcal{S}\}.
\]
Choose any $T\in\mathcal{S}$ with $\MM{T}=D^*$. Then
\[
E_s(T)\le \bin{k}(D^*)^{-s}.
\]
Now let
\[
R:=\max\{\MM{S}:S\in\mathcal{S},\ \MM{S}<D^*\}.
\]
Because $\mathcal{S}$ is finite, we have $R<D^*$. For every $S\in\mathcal{S}$ with $\MM{S}<D^*$, some selected pair has distance at most $R$, hence
\[
E_s(S)\ge R^{-s}.
\]
Choose $s_0$ large enough that
\[
R^{-s}>\bin{k}(D^*)^{-s}
\qquad\text{for all }s\ge s_0.
\]
Then for all such $s$, every subset with $\MM{S}<D^*$ has strictly larger Riesz energy than $T$. Therefore every minimizer of $E_s$ must satisfy $\MM{S}=D^*$, that is, it must be MPD optimal.
\end{proof}

\begin{remark}
This known finite-instance observation is central to our comparison with minimum pairwise distance (MPD). It is interpretive rather than a new uniform complexity result: it explains why many large-$s$ phenomena naturally converge to the classical MPD picture.
\end{remark}

\section{Open problems and boundary of the tractable regime}

The preceding theorems draw a sharp boundary, but not a complete classification. The remaining open cases are structurally meaningful rather than incidental.

\paragraph{Fixed-exponent Euclidean hardness.}
Our Euclidean hardness proof uses $s$ as part of the input. Extending NP-hardness to every fixed exponent $s>0$ in low-dimensional Euclidean spaces remains open. This is arguably the most natural missing theorem on the hard side. Appendix~\ref{app:fowler} records a Fowler-style fixed-$s$ route and explains precisely why the present admissible-versus-forbidden threshold estimates do not yet close such a reduction.

\paragraph{One-dimensional and ordered Euclidean inputs.}
The line (and also $\ell_1 ordered staircases$, see \cite{Emmerich2026staircases}) appears significantly more structured than arbitrary Euclidean instances, but one should be careful not to overinterpret this. The contrast with minimum pairwise distance (MPD) is already sharp in low Euclidean dimension. For the MPD objective on sorted points $x_1<\dots<x_n$, an exact Bellman-style recurrence can be written in terms of the last chosen point:
\[
M(i,t)=\max_{j<i}\min\{M(j,t-1),\,x_i-x_j\},
\]
where $M(i,t)$ denotes the best possible bottleneck value of a $t$-subset ending at $x_i$. Thus the MPD optimization problem is polynomial-time solvable on the line by dynamic programming, whereas in dimension two it is already NP-hard via the same geometric independent-set machinery that underlies our plane hardness discussion \cite{agarwal2006computing}. By contrast, no analogous left/right state compression works for minimum Riesz $s$-energy: when a new point $x_i$ is added, the extra cost is
\[
\sum_{x\in S}|x_i-x|^{-s},
\]
which depends on the full internal configuration of the already selected subset $S$, not merely on its cardinality or its extreme points. In particular, the naive Bellman principle that suffices for MPD on the line breaks down for Riesz energy (see \cite{Emmerich2025}, where also a concrete counterexample is provided). Thus the one-dimensional Euclidean case remains open here for a substantive structural reason, rather than because of a missing implementation detail. 

Notably, however, the Riesz $s$-energy on ordered point sets exhibits considerable positive structure.
For example, it is straightforward to show that it satisfies a Monge
condition, which can be exploited to obtain algorithmic speedups and
certificates for multilevel dynamic programming. This may be useful, for
instance, when considering higher-order neighborhoods in the Bellman
recursion. A detailed discussion of these aspects, however, exceeds the
scope of the present paper and is left for future work.

In the limit $s \rightarrow \infty$, the picture becomes clearer, because the limiting maximum minimum pairwise distance (MPD) problem on the line admits a polynomial time 
algorithm based on a simple greedy feasibility test. 
To test whether there exists a \(k\)-subset
with MPD at least \(\tau\), start with the leftmost point and then scan the
points from left to right, always selecting the first point whose distance
from the last selected point is at least \(\tau\). At the end of the scan,
one checks whether at least \(k\) points have been selected. This greedy test
is exact: choosing each feasible point as far to the left as possible cannot
reduce the remaining feasible space to its right.
The optimal MPD value is one of the pairwise distances \(x_j-x_i\),
\(1\le i<j\le n\). Hence an exact implementation can explicitly form and sort
all pairwise distances and then apply binary search with the greedy
feasibility test. This gives an \(O(n^2\log n)\) time bound and \(O(n^2)\)
storage. Alternatively, one may apply the Frederickson--Johnson selection and
ranking algorithm for sorted matrices to the implicit sorted matrix of
pairwise differences~\cite{frederickson1984generalized}. This avoids
materializing the \(O(n^2)\) pairwise distances and yields an
\(O(n\log n)\)-type exact bound after the input points have been sorted, but
requires the nontrivial sorted-matrix-search machinery.

\paragraph{Parameterized complexity.}
The problem also invites parameterized analysis in $k$, dimension, or structural parameters of thresholded interaction graphs. We leave this aside here in order to keep the paper focused on the primary hard-versus-tractable dichotomy. Please refer to \cite{Atta2026} for a first result in the direction of PTAS constructions.  

\section{Complexity landscape}

The results of this paper isolate a compact but informative landscape for minimum Riesz $s$-energy subset selection. On the hard side, the problem is NP-hard already on finite metric spaces for every fixed $s>0$, and it remains NP-hard in the Euclidean plane --- hence in every fixed Euclidean dimension $d\ge 2$ --- when the exponent is part of the input. Thus neither unrestricted metric structure nor low-dimensional Euclidean geometry is, by itself, enough to make the problem easy.

On the tractable side, finite ultrametric spaces admit exact dynamic programming because the interaction between two child subproblems depends only on the selected cardinalities and the common cross-subtree distance. This identifies ultrametrics as a genuine tractable island rather than a degenerate case. For the comparison with minimum pairwise distance (MPD), a central ingredient is the known finite-instance large-$s$ observation that every Riesz $s$-energy minimizer is eventually MPD optimal. This explains why large-$s$ behavior increasingly resembles bottleneck dispersion even though the finite-$s$ objective is globally additive over all pairs.

The low-dimensional Euclidean comparison is particularly revealing. For minimum pairwise distance (MPD) on the line, the ordered structure supports a Bellman-style recurrence, whereas for Riesz $s$-energy the same state compression fails because adding a new point contributes a sum over all previously selected points. Thus the one-dimensional Riesz case remains open for a structural reason, not merely because the right implementation has not yet been found. Table~\ref{tab:landscape} summarizes this picture.

\begin{table}[t]
\centering
\small
\begin{tabular}{p{5cm}p{5cm}p{4.3cm}}
\toprule
Setting & Minimum Riesz $s$-energy subset selection & Large-$s$ / minimum-pairwise-distance viewpoint \\
\midrule
General finite metric spaces & NP-hard for every fixed $s>0$ (Theorem~\ref{thm:metric}); general-metric hardness already appears in \cite{PereverdievaEtAl2025} & NP-hard \cite{agarwal2006computing} \\
Euclidean plane, $s$ in input & NP-hard (Theorem~\ref{thm:gis}) & Consistent with the fact that the MPD optimization problem is already NP-hard in the plane \\
Euclidean plane, fixed exponent $s$  & Open in this paper; for $s>2$, Appendix~\ref{app:fowler} gives an $O(k)$ far-field bound for Fowler-style gadgets. & NP-hard in 2-D \cite{agarwal2006computing} \\
Finite ultrametric spaces & Exact DP in $O(nk^2)$  (Theorem~\ref{thm:ultra}, Corollary~\ref{cor:ultra-complexity}) &  Exact DP in $O(nk^2)$; comparison uses the known finite-instance large-$s$ link to MPD optimality (Proposition~\ref{thm:limit}) \\
One-dimensional Euclidean inputs & Open in this paper; left--right DP is not exact for Riesz energy. & Polynomial-time solvable for MPD, e.g., by a Bellman-style left--right dynamic program;  \\
\bottomrule
\end{tabular}
\caption{Summary of the hard and tractable cases isolated in this paper.}
\label{tab:landscape}
\end{table}

\section{Conclusion}
Minimum Riesz $s$-energy subset selection exhibits a concise and, to our
knowledge, previously unrecorded complexity picture. Building on the earlier
general-metric hardness context established in the comparative-analysis work of
Pereverdieva et al.~\cite{PereverdievaEtAl2025}, we showed that the problem
remains hard in the Euclidean plane when the exponent $s$ is part of the input.
At the same time, finite ultrametric spaces admit exact dynamic programming,
showing that hierarchical metrics recover nontrivial algorithmic structure.
Together with the known finite-instance large-$s$ link to minimum pairwise distance (MPD),
this shows that the problem naturally interpolates between pairwise repulsion
and bottleneck-style separation. This makes the boundary between general
geometry, the one-dimensional ordered setting, and hierarchical structure
especially interesting: ultrametrics are exactly solvable, the Euclidean plane
is hard, and the one-dimensional case remains structurally subtle because the
natural state compression that works for MPD does not directly transfer to
Riesz $s$-energy. Beyond exact complexity, recent work also points toward the need to explore practical
algorithmic directions, including fast Riesz-based subset-selection heuristics
in multiobjective optimization and approximation results for related geometric
facility-location variants~\cite{FalconCardonaEtAl2025,Emmerich2025,Atta2026}.

\appendix

\section{Appendix A: Dynamic programming on finite ultrametric spaces}\label{app:ultra}

This appendix gives a fuller discussion of the ultrametric dynamic programming result stated in the main text. We first derive the recurrence, then formulate the corresponding algorithm explicitly, and finally illustrate it on a small abstract tree example and report a brief numerical validation.

\subsection{Derivation of the recurrence}

Let $u$ be an internal node of a rooted binary ultrametric tree with children $v$ and $w$. Denote by $L(u)$ the leaves below $u$, and let $\Delta_u$ be the common distance between every leaf in $L(v)$ and every leaf in $L(w)$. If $A\subseteq L(v)$ and $B\subseteq L(w)$, then every cross pair contributes the same amount $\Delta_u^{-s}$ to the Riesz energy. Therefore
\[
E_s(A\cup B)=E_s(A)+E_s(B)+|A||B|\,\Delta_u^{-s}.
\]
This identity is the only nontrivial structural ingredient. Once it is available, the Bellman recurrence follows immediately by minimizing over all splits of a target cardinality.

\begin{proof}[Proof of Theorem~\ref{thm:ultra}]
Any subset $S\subseteq L(u)$ of size $t$ decomposes uniquely as $S=A\cup B$ with $A\subseteq L(v)$ and $B\subseteq L(w)$. Every pair inside $A$ contributes to $E_s(A)$, every pair inside $B$ contributes to $E_s(B)$, and every cross pair contributes the same amount $\Delta_u^{-s}$. Since there are $|A||B|$ cross pairs, we obtain
\[
E_s(S)=E_s(A)+E_s(B)+|A||B|\,\Delta_u^{-s}.
\]
Minimizing over all splits $t=|A|+|B|$ yields the stated recurrence. The leaf initialization is immediate from the definition of Riesz $s$-energy.
\end{proof}

\begin{proof}[Proof of Corollary~\ref{cor:ultra-complexity}]
The recurrence is evaluated bottom-up. For each internal node and each target cardinality $t\le k$, all splits $t_v+t_w=t$ are tried. This is the same counting as in a knapsack convolution, hence $O(k^2)$ work per internal node and $O(nk^2)$ total work over the tree. Storing one table of size $O(k)$ per node gives $O(nk)$ space.
\end{proof}

\subsection{Algorithmic form}

Algorithm~\ref{alg:ultra-dp} records the resulting exact dynamic program.  A fully documented reference implementation and validation suite for this
algorithm is available at
\url{https://github.com/emmerichmtm/RieszSEnergyUltrametricsDP}.

\begin{algorithm}[h]
\caption{Exact DP for minimum Riesz $s$-energy subset selection on a rooted binary ultrametric tree}
\label{alg:ultra-dp}
\begin{algorithmic}[1]
\Require Rooted binary ultrametric tree $T$, target size $k$, exponent $s>0$
\ForAll{leaves $\ell$ of $T$}
    \State set $F_\ell(0)\gets 0$ and $F_\ell(1)\gets 0$
\EndFor
\ForAll{internal nodes $u$ of $T$ in postorder}
    \State let $v,w$ be the children of $u$
    \For{$t=0,1,\dots,k$}
        \State initialize $F_u(t)\gets +\infty$
        \For{$t_v=0,1,\dots,t$}
            \State $t_w\gets t-t_v$
            \State $cand\gets F_v(t_v)+F_w(t_w)+t_v t_w\,\Delta_u^{-s}$
            \State $F_u(t)\gets \min\{F_u(t),cand\}$
        \EndFor
    \EndFor
\EndFor
\State \Return the root value $F_r(k)$
\end{algorithmic}
\end{algorithm}

\subsection{Worked abstract tree example}

We illustrate the recurrence on a small abstract rooted tree with six taxa. The tree consists of two cherries, $\{a,b\}$ and $\{c,d\}$, together with two further taxa $e$ and $f$. The first cherry merges at height $3.5$, the second at height $4$, and the two resulting groups merge at height $5.5$. Figure~\ref{fig:primate-tree-app} shows the corresponding stylized ultrametric tree. Such numbers may be interpreted abstractly as times in a phylogenetic tree, as in biodiversity applications, or as cophenetic heights in a hierarchical clustering of cultural vectors~\cite{HartmannSteel2006,StivalaEtAl2014}.

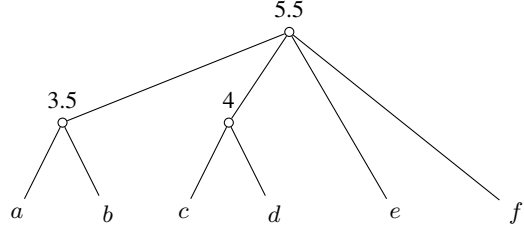
\begin{figure}[h]
\centering
\begin{tikzpicture}[
  every node/.style={font=\small},
  time/.style={circle, draw, inner sep=1.2pt, fill=white}
]
\node[time,label=above:{5.5}] (r) at (0,2.4) {};
\node[time,label=above:{3.5}] (l1) at (-3,1.2) {};
\node[time,label=above:{4}] (l2) at (-0.8,1.2) {};
\node[draw=none] (e) at (1.4,0) {$e$};
\node[draw=none] (f) at (3.0,0) {$f$};
\draw (r)--(l1);
\draw (r)--(l2);
\draw (r)--(e);
\draw (r)--(f);
\node[draw=none] (a) at (-3.6,0) {$a$};
\node[draw=none] (b) at (-2.4,0) {$b$};
\node[draw=none] (c) at (-1.4,0) {$c$};
\node[draw=none] (d) at (-0.2,0) {$d$};
\draw (l1)--(a);
\draw (l1)--(b);
\draw (l2)--(c);
\draw (l2)--(d);
\end{tikzpicture}
\caption{An abstract ultrametric tree on six taxa with two cherries. The pairs $\{a,b\}$ and $\{c,d\}$ form cherries, while $e$ and $f$ are attached directly at the top level. Distances between leaves are determined by the height of the least common ancestor. The numerical labels are illustrative only.}
\label{fig:primate-tree-app}
\end{figure}

If leaf-to-leaf distances are taken as twice the least-common-ancestor heights, then
\[
d(a,b)=7,
\qquad
 d(c,d)=8,
\qquad
 d(x,y)=11
\]
for all other distinct pairs $x,y$. Take $s=1$. Then the pair $\{a,b\}$ has energy $1/7$, the pair $\{c,d\}$ has energy $1/8$, and every pair drawn from different top-level branches has energy $1/11$.
For three-taxon subsets, one obtains for instance
\[
E_1(\{a,c,e\})=\frac{1}{11}+\frac{1}{11}+\frac{1}{11}=\frac{3}{11}\approx 0.272727,
\]
whereas
\[
E_1(\{a,b,e\})=\frac{1}{7}+\frac{1}{11}+\frac{1}{11}=\frac{25}{77}\approx 0.324675.
\]
Thus the energy-minimizing three-taxon subsets avoid the closer cherry pair $\{a,b\}$ whenever possible, exactly as the recurrence predicts.

\subsection{Numerical validation and implementation}

To validate the recurrence computationally, we implemented the ultrametric DP and compared its output against brute-force enumeration on a deterministic set of small random rooted binary ultrametric trees. On all tested instances the dynamic-programming optimum matched the brute-force optimum up to floating-point roundoff. A small reference Python implementation accompanies the present work See: \url{https://github.com/emmerichmtm/RieszSEnergyUltrametricsDP}.

\section{Appendix B: A Fowler-style fixed-\texorpdfstring{$s$}{s} route in the Euclidean plane and its limitation}
\label{app:fowler}

This appendix records an alternative route toward Euclidean hardness at a fixed exponent $s$, inspired by the classical gadget construction of Fowler, Paterson, and Tanimoto~\cite{FowlerPatersonTanimoto1981} and translated here from unit squares to unit discs. The aim is not to claim a complete fixed-$s$ hardness proof, but to document a geometrically natural approach and to isolate precisely where the present estimates stop short. At the same time, the analysis yields a useful structural fact: in the relevant gadget class, the admissible far-field Riesz energy is only \emph{linear} in the number of selected points, rather than quadratic as a naive pair-count estimate would suggest.

\subsection{Geometric idea}

The intended reduction route is directly from \textsc{3SAT}. As in Fowler et al., one arranges wire gadgets with two states, crossover gadgets preserving the transmitted state, and clause gadgets that allow one additional choice if and only if a literal is satisfied. In the present context, the choices are represented by points in the Euclidean plane, or equivalently by equal-radius discs centred at those points.

The geometric idea is simple. A forbidden local configuration forces two selected discs to come too close, creating a very short centre-to-centre distance and therefore a large Riesz contribution. One may therefore hope to separate satisfying from unsatisfying assignments by making every inconsistent choice incur such an overlap penalty, while every consistent choice incurs only smaller contributions from more distant interactions.

Figures~\ref{fig:boolean-circuit-app}--\ref{fig:unit-disk-gadgets-app} reproduce the figure structure from the longer manuscript and illustrate this idea.

\begin{figure}[h]
{
  \begin{center}
  \begin{minipage}[b]{0.4\linewidth}
    \centering
    \resizebox{\linewidth}{!}{\TikzBooleanCircuit}
    \caption{Boolean circuit gadget.}\label{fig:boolean-circuit-app}
  \end{minipage}
  \hfill
  \begin{minipage}[b]{0.4\linewidth}
    \centering
    \resizebox{\linewidth}{!}{\TikzWireGadget}
    \caption*{Wire gadget.}
  \end{minipage}

  \vspace{1em}

  \begin{minipage}[b]{0.4\linewidth}
    \centering
    \resizebox{\linewidth}{!}{\TikzCrossoverGadget}
    \caption*{Crossover gadget.}
  \end{minipage}
  \hfill
  \begin{minipage}[b]{0.4\linewidth}
    \centering
    \resizebox{\linewidth}{!}{\TikzClauseGadget}
    \caption*{Clause gadget.}
  \end{minipage}
  \end{center}
}
\caption{Disk-based gadget layout in the Euclidean plane following the structure of Fowler, Paterson, and Tanimoto~\cite{FowlerPatersonTanimoto1981}: a Boolean circuit encoding of a 3SAT instance, wire gadgets with two states, crossover gadgets preserving the transmitted state, and clause gadgets in which one extra disk fits if and only if a literal is satisfied. In a full fixed-$s$ proof attempt, admissible selections would correspond to non-overlapping choices consistent with a satisfying assignment.}
\label{fig:unit-disk-gadgets-app}
\end{figure}

We arrange the disc centres so that every admissible selected pair is at distance at least $2r$, while any forced conflict creates a pair at distance at most $\tfrac{3}{2}r$. Thus a forbidden overlap contributes at least
\[
\left(\tfrac{3}{2}r\right)^{-s},
\]
whereas an admissible pair contributes at most
\[
(2r)^{-s}.
\]
This produces a constant local gap.

\begin{lemma}[Local overlap gap]
\label{lem:gap-app}
At exponent $s>0$, any forced-overlap pair of selected discs whose centres have distance at most $\tfrac{3}{2}r$ contributes at least
\[
\left(\tfrac{3}{2}r\right)^{-s}
\]
to the Riesz energy, while any admissible non-overlap pair whose centres have distance at least $2r$ contributes at most
\[
(2r)^{-s}.
\]
\end{lemma}

\begin{proof}
This is immediate from monotonicity of the map $d\mapsto d^{-s}$ for $s>0$.
\end{proof}

Lemma~\ref{lem:gap-app} explains why this route is attractive: forbidden local configurations are penalized more heavily than admissible ones. To obtain a reduction, however, one must compare a single such overlap penalty with the \emph{entire} energy of an admissible $k$-point selection. This is where the global budget enters.

\subsection{Packing-based far-field budget in the Euclidean plane}

To bound the admissible side globally, we use the standard packing estimate for equal-radius discs in the Euclidean plane. Figure~\ref{fig:discs-app} is the same packing picture used in the longer draft.

\begin{figure}[h]
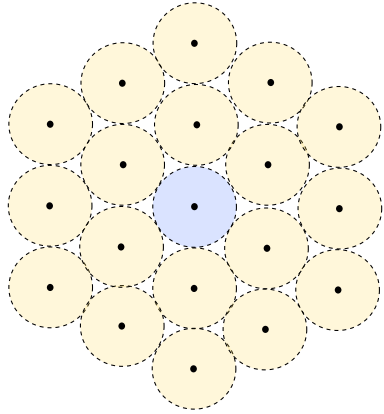

\centering
\resizebox{0.3\linewidth}{!}{\TikzHexagonalPacking}
\caption{Packing of equal-radius discs in the Euclidean plane. Around a fixed selected disc, the number of non-overlapping discs in the $i$th layer grows at most linearly in $i$. This leads to a convergent sum for every fixed $s>2$.}
\label{fig:discs-app}
\end{figure}

Fix one selected disc with centre $p$. Around this disc, consider the concentric annuli
\[
A_i(p)=\bigl\{x\in\mathbb{R}^2 : 2ir \le \|x-p\| < 2(i+1)r\bigr\},
\qquad i\ge 1.
\]
Because the selected discs are non-overlapping, only linearly many other disc centres can lie in the $i$th annulus.

\begin{lemma}[Pointwise far-field budget]
\label{lem:pointwise-budget-app}
Fix $s>2$. Let $p$ be the centre of a selected disc in a non-overlapping configuration of equal-radius discs with radius $r$, and let $\zeta$ denote the Riemann zeta function. Then
\[
\sum_{q\in S\setminus\{p\}} \|p-q\|^{-s}
\le
6(2r)^{-s}\sum_{i\ge 1} i^{1-s}
=
6(2r)^{-s}\zeta(s-1).
\]
In particular, for every fixed $s>2$, the total contribution of all other selected points to one selected point is bounded by a constant depending only on $s$ and on the geometric scale $r$.
\end{lemma}

\begin{proof}
In the Euclidean plane, at most $6i$ non-overlapping equal-radius discs can have centres in the $i$th annulus around the centre $p$ of a fixed selected disc. Every such centre is at distance at least $2ir$ from $p$, hence each such point contributes at most
\[
(2ir)^{-s}.
\]
Therefore the total contribution from the $i$th annulus is at most
\[
6i\,(2ir)^{-s}=6(2r)^{-s}i^{1-s}.
\]
Summing over all annuli gives
\[
\sum_{q\in S\setminus\{p\}} \|p-q\|^{-s}
\le
6(2r)^{-s}\sum_{i\ge 1} i^{1-s}
=
6(2r)^{-s}\zeta(s-1),
\]
which converges for every fixed $s>2$.
\end{proof}

\begin{remark}[Basel problem special case]
\label{rem:basel-app}
For the concrete case $s=3$, Lemma~\ref{lem:pointwise-budget-app} becomes
\[
\sum_{q\in S\setminus\{p\}} \|p-q\|^{-3}
\le
6(2r)^{-3}\sum_{i\ge 1} i^{-2}
=
6(2r)^{-3}\zeta(2)
=
6(2r)^{-3}\frac{\pi^2}{6}.
\]
Thus the classical Basel value $\zeta(2)=\pi^2/6$ appears naturally as the packing-based bound on the total contribution to one selected point.
\end{remark}

The pointwise estimate immediately implies a linear global bound.

\begin{corollary}[Linear admissible-energy budget]
\label{cor:linear-budget-app}
For every fixed $s>2$, there exists a constant $C_s>0$ such that every admissible non-overlapping $k$-point selection $S$ in this gadget class satisfies
\[
E_s(S)\le C_s\,k.
\]
\end{corollary}

\begin{proof}
Write the unordered Riesz energy as
\[
E_s(S)
=
\frac12\sum_{p\in S}\sum_{q\in S\setminus\{p\}}\|p-q\|^{-s}.
\]
By Lemma~\ref{lem:pointwise-budget-app}, each inner sum is bounded by the same constant
\[
6(2r)^{-s}\zeta(s-1).
\]
Hence
\[
E_s(S)\le
\frac12\,k\cdot 6(2r)^{-s}\zeta(s-1),
\]
which proves the claim.
\end{proof}

\begin{remark}
Corollary~\ref{cor:linear-budget-app} is the key structural message of this appendix. A naive global estimate would bound the admissible energy by $O(k^2)$ merely because there are $\binom{k}{2}$ pairs. The packing argument shows that this is far too crude: in the present Euclidean-plane gadget class, the admissible far-field energy is only $O(k)$.
\end{remark}

\subsection{Why this still does not close a fixed-$s$ reduction}

Suppose now that a selected set contains at least one forbidden overlap pair. Then by Lemma~\ref{lem:gap-app} its energy is at least
\[
\left(\tfrac{3}{2}r\right)^{-s},
\]
plus the contributions of all remaining pairs. On the admissible side, Corollary~\ref{cor:linear-budget-app} gives the upper bound
\[
E_s^{\mathrm{adm}}(k)\le C_s\,k.
\]
Thus the overlap penalty must be compared not with a quadratic admissible budget, but with a linear one.

\begin{proposition}[Interpretation of the obstruction]
\label{prop:fowler-obstruction-app}
The Fowler-style route in the Euclidean plane shows that the obstruction to a fixed-$s$ reduction is not a quadratic explosion of admissible pair contributions. Rather, the admissible-energy budget is only linear in $k$, but still grows with $k$. Therefore a uniform reduction at fixed $s$ would require overlap penalties strong enough to dominate that linear growth.
\end{proposition}

\begin{proof}
By Corollary~\ref{cor:linear-budget-app}, every admissible $k$-point selection has energy at most $C_s k$. A single forbidden overlap contributes at least $\left(\tfrac{3}{2}r\right)^{-s}$ by Lemma~\ref{lem:gap-app}. For a reduction based solely on this gap, one would need the latter quantity to dominate the former uniformly in $k$. Since the admissible budget still grows linearly with $k$, the present estimates do not by themselves yield a fixed-$s$ reduction.
\end{proof}

This is the main point of the appendix. The failure of the fixed-$s$ route is subtler than a naive pair-count argument would suggest. The obstruction is \emph{not} that admissible long-range interactions explode quadratically. On the contrary, packing controls the per-point contribution by a constant and hence the total admissible energy linearly in $k$. The real issue is that even this sharpened linear budget must still be dominated by the penalty of a forbidden overlap.

This also explains why the successful Euclidean hardness proof lets the exponent depend on $k$. Increasing $s$ amplifies the ratio between the forbidden short distance $\tfrac{3}{2}r$ and the admissible distance $2r$, so that one forbidden pair can dominate the entire linear admissible-energy budget.

\begin{remark}
The value of the Fowler-style route is therefore twofold. First, it provides a natural geometric mechanism for encoding inconsistency by overlap penalties in the Euclidean plane. Second, it clarifies why the successful reduction uses an exponent depending on $k$: not because the admissible far-field energy is uncontrolled, but because even its sharp linear bound still grows with $k$.
\end{remark}


\begin{thebibliography}{99}

\bibitem{AgarwalVanKreveldSuri1998}
P. K. Agarwal, M. van Kreveld, and S. Suri.
\newblock Label placement by maximum independent set in rectangles.
\newblock \emph{Computational Geometry}, 11(3--4):209--218, 1998.

\bibitem{agarwal2006computing}
P.~K.~Agarwal, M.~Overmars, and M.~Sharir, Computing maximally separated sets in the plane,
\emph{SIAM Journal on Computing},
36(3):815--834, 2006.

\bibitem{Atta2026}
S. Atta.
\newblock Riesz Energy Minimization Facility Location Problem in the Plane:
\newblock Complexity and a Polynomial-Time Approximation Scheme.
\newblock {\em Theoretical Computer Science}, 1070:115833, 2026.
\newblock doi:10.1016/j.tcs.2026.115833.

\bibitem{clark1990unit}
B.~N.~Clark, C.~J.~Colbourn, and D.~S.~Johnson,
``Unit disk graphs,''
\emph{Discrete Mathematics},
86(1--3):165--177, 1990.



\bibitem{BorodachovHardinSaff2019}
S. V. Borodachov, D. P. Hardin, and E. B. Saff.
\newblock {\em Discrete Energy on Rectifiable Sets}.
\newblock Springer Monographs in Mathematics.
\newblock Springer, New York, 2019.
\newblock doi:10.1007/978-0-387-84808-0.


\bibitem{Emmerich2026staircases} Emmerich, M. (2026). Exact Dynamic Programming for Solow--Polasky Diversity Subset Selection on Lines and Staircases. arXiv:2604.26929.

\bibitem{ArenasMerklPichlerRiveros2024}
M. Arenas, T.~C. Merkl, R. Pichler, and C. Riveros.
Towards Tractability of the Diversity of Query Answers: Ultrametrics to the Rescue.
\emph{Proceedings of the ACM on Management of Data}, 2(5):215:1--215:26, 2024.

\bibitem{BrauchartGrabner2015}
J. S. Brauchart and P. J. Grabner.
\newblock Distributing many points on spheres: minimal energy and designs.
\newblock \emph{Journal of Complexity}, 31(3):293--326, 2015.
\newblock doi:10.1016/j.jco.2015.02.001.

\bibitem{ClarkColbournJohnson1990}
B. N. Clark, C. J. Colbourn, and D. S. Johnson.
\newblock Unit disk graphs.
\newblock \emph{Discrete Mathematics}, 86(1--3):165--177, 1990.

\bibitem{Emmerich2025}
M. T. M. Emmerich.
\newblock Minimum Riesz $s$-Energy Subset Selection in Ordered Point Sets via Dynamic Programming.
\newblock arXiv:2502.01163, 2025.

\bibitem{frederickson1984generalized}
G.~N. Frederickson and D.~B. Johnson.
\newblock Generalized selection and ranking: Sorted matrices.
\newblock \emph{SIAM Journal on Computing}, 13(1):14--30, 1984.
\newblock doi: \href{https://doi.org/10.1137/0213002}{10.1137/0213002}.

\bibitem{FalconCardonaUribeRosas2024}
J. G. Falcón-Cardona, L. Uribe, and P. Rosas.
\newblock Riesz $s$-Energy as a Diversity Indicator in Evolutionary Multi-Objective Optimization.
\newblock {\em IEEE Transactions on Evolutionary Computation}, early online, 2024.
\newblock doi:10.1109/TEVC.2024.3405197.

\bibitem{FalconCardonaEtAl2025}
J. G. Falcón-Cardona, J. Juárez, L. A. Márquez-Vega, and M. T. M. Emmerich.
\newblock Fast High-Diversity Subset Selection for Multi-Objective Optimization by Riesz $s$-Energy.
\newblock {\em IEEE Transactions on Evolutionary Computation}, early online, 2025.
\newblock doi:10.1109/TEVC.2025.3570938.

\bibitem{FowlerPatersonTanimoto1981}
R. J. Fowler, M. S. Paterson, and S. L. Tanimoto.
\newblock Optimal packing and covering in the plane are NP-complete.
\newblock \emph{Information Processing Letters}, 12(3):133--137, 1981.

\bibitem{FejesToth1942}
L. Fejes T\'oth.
\newblock {\"U}ber die dichteste Kugellagerung.
\newblock \emph{Mathematische Zeitschrift}, 48:676--684, 1943.

\bibitem{ghosh1996computational}
J. B. Ghosh.
\newblock Computational aspects of the maximum diversity problem.
\newblock {\em Operations Research Letters}, 19(4):175--181, 1996.
\newblock doi:10.1016/0167-6377(96)00025-9.

\bibitem{GlazkoNei2003}
L. A. Glazko and M. Nei.
\newblock Estimation of divergence times for major lineages of primate species.
\newblock \emph{Molecular Biology and Evolution}, 20(3):424--434, 2003.
\newblock doi:10.1093/molbev/msg050.

\bibitem{HartmannSteel2006}
K. Hartmann and M. Steel.
\newblock Maximizing phylogenetic diversity in biodiversity conservation: greedy solutions to the Noah's Ark problem.
\newblock \emph{Systematic Biology}, 55(4):644--651, 2006.
\newblock doi:10.1080/10635150600873876.

\bibitem{MansonSempleSteel2022}
K. Manson, C. Semple, and M. Steel.
Counting and optimising maximum phylogenetic diversity sets.
\emph{Journal of Mathematical Biology}, 85:11, 2022.
doi:10.1007/s00285-022-01779-3.


\bibitem{MinhKlaereHaeseler2006}
B.~Q. Minh, S. Klaere, and A. von Haeseler.
Phylogenetic Diversity within Seconds.
\emph{Systematic Biology}, 55(5):769--773, 2006.
doi:10.1080/10635150600981604.

\bibitem{PereverdievaEtAl2025}
K. Pereverdieva, A. Deutz, T. Ezendam, T. B\"ack, H. Hofmeyer, and M. T. M. Emmerich.
\newblock Comparative analysis of indicators for multi-objective diversity optimization.
\newblock In \emph{Evolutionary Multi-Criterion Optimization}, volume 15513 of \emph{Lecture Notes in Computer Science}, pages 58--71. Springer, 2025.
\newblock doi:10.1007/978-981-96-3538-2\_5.

\bibitem{RaviRosenkrantzTayi1994}
S. S. Ravi, D. J. Rosenkrantz, and G. K. Tayi.
\newblock Heuristic and special case algorithms for dispersion problems.
\newblock {\em Operations Research}, 42(2):299--310, 1994.
\newblock doi:10.1287/opre.42.2.299.

\bibitem{StivalaEtAl2014}
A. Stivala, G. Robins, Y. Kashima, and M. Kirley.
\newblock Ultrametric distribution of culture vectors in an extended Axelrod model of cultural dissemination.
\newblock \emph{Scientific Reports}, 4:4870, 2014.
\newblock doi:10.1038/srep04870.





\end{thebibliography}
\end{document}